\documentclass[runningheads]{llncs}

\usepackage{amssymb, amsmath, amsfonts, mathtools}                                                  
\usepackage{microtype}                                                                              
\usepackage[linesnumbered, ruled, vlined, procnumbered, nokwfunc]{algorithm2e}                      
\usepackage{booktabs, tabularx}                                                                     
\usepackage{environ}                                                                                
\usepackage{hyperref}                                                                               
\usepackage{soul, color}                                                                            
\usepackage[dvipsnames, table]{xcolor}                                                              
\usepackage{pifont}                                                                                 
\usepackage[normalem]{ulem}                                                                         
\usepackage{float}                                                                                  
\usepackage{adjustbox}                                                                              
\usepackage{pdflscape}                                                                              
\usepackage{textgreek}                                                                              

\hypersetup{
  colorlinks = true,
  linktoc = all,
  citecolor = RoyalBlue!95!Black,
  linkcolor = RedOrange!85!Black,
  urlcolor = PineGreen!95!Black,
  filecolor = cyan, 
}

\newcommand*{\prob}[1]{\textsc{#1}}                                                                 
\newcommand*{\probc}[1]{\textsf{#1}}                                                                
\newcommand*{\algo}[1]{\texttt{#1}}                                                                 
\newcommand*{\algoi}[2]{\texttt{#1}(#2)}                                                            
\SetKw{Read}{read}
\SetKw{Write}{write}
\SetKw{Break}{break}
\SetKw{Continue}{continue}
\SetAlgoCaptionSeparator{,}

\DeclarePairedDelimiter{\brnX}{(}{)}
\DeclarePairedDelimiter{\brsX}{[}{]}

\DeclarePairedDelimiter{\brcX}{\{}{\}}
\DeclarePairedDelimiter{\absX}{\lvert}{\rvert}

\DeclarePairedDelimiter{\ceilX}{\lceil}{\rceil}

\newcommand*{\brn}{\brnX*}                                                                          
\newcommand*{\brs}{\brsX*}                                                                          
\newcommand*{\brc}{\brcX*}                                                                          
\newcommand*{\abs}{\absX*}                                                                          
\newcommand*{\ceil}{\ceilX*}                                                                        

\DeclareMathOperator{\Pop}{P}

\newcommand*{\Prs}[2]{\Pop_{#1}\brn{#2}}                                                            

\DeclareMathOperator{\ohop}{O}

\newcommand*{\Oh}[1]{\ohop\brn{#1}}                                                                 

\DeclareMathOperator{\Vop}{V}

\newcommand*{\V}[1]{\Vop\brn{#1}}                                                                   
\let\degd\deg
\let\deg\relax
\newcommand*{\deg}[1]{\degd\brn{#1}}                                                                

\newcommand*{\N}{\mathbb{N}}                                                                        
\newcommand*{\Z}{\mathbb{Z}}                                                                        
\newcommand*{\setb}[2]{\left\{#1 \mid #2\right\}}                                                   

\newcommand*{\dnimply}{\kern.6em\not\kern-.6em \implies}                                            

\newcommand*{\sgn}[1]{\sgn\brn{#1}}                                                                 
\newcommand*{\oper}[1]{\operatorname{#1}}                                                           

\makeatletter
\newcommand{\ptitle}[1]{\gdef\prob@title{#1}}                                                       
\newcommand{\pobject}[1]{\gdef\prob@object{#1}}                                                     
\newcommand{\pquery}[1]{\gdef\prob@query{#1}}                                                       
\newcommand{\pparam}[1]{\gdef\prob@param{#1}}                                                       

\NewEnviron{problemx}[1][]{%
  \def\prob@type{#1}
  \def\prob@topt{opt}%
  \def\prob@title{}%
  \def\prob@object{}%
  \def\prob@query{}%
  \def\prob@param{}%
  \ifx\prob@type\prob@topt
    \def\prob@qword{Solution}%
  \else
    \def\prob@qword{Question}%
  \fi%
  \BODY
  \ifx\prob@param\@empty
    \def\prob@trow{{\large\prob@title}}
  \else%
    \def\prob@trow{
      \begin{tabular*}{\textwidth}{@{\extracolsep{\fill}}lr}%
        {\large\prob@title} & \textbf{Parameter:}~\prob@param%
      \end{tabular*}%
    }%
  \fi%
  \par\noindent\ignorespaces%
  \ignorespacesafterend%
  \\*[-1ex]%
    \begin{minipage}{\textwidth}%
      \vspace{1.2ex}
      \prob@trow\\[2\fboxsep]
      \textbf{Instance:}~\prob@object\\
      \textbf{\prob@qword:}~\prob@query%
    \end{minipage}%
  \bigskip\par%
}
\makeatother

\title{Sublinear-Space Approximation Algorithms for Max $r$-SAT}
\titlerunning{Sublinear-Space Max $r$-SAT Approximation} 
\institute{The Institute of Mathematical Sciences, HBNI, Chennai, India\\\email{\{barindam,vraman\}@imsc.res.in}}
\author{Arindam Biswas \and Venkatesh Raman}
\authorrunning{A.\ Biswas and V.\ Raman}

\newcommand*{\pMaxrSAT}{\prob{Max \textnormal{\textit{r}}-SAT}}
\newcommand*{\pPlMaxrSAT}{\prob{Planar Max \textnormal{\textit{r}}-SAT}}

\newcommand*{\pLP}{\prob{Linear Programming}}

\newcommand*{\cP}{\probc{P}}
\newcommand*{\cNP}{\probc{NP}}

\begin{document}
\maketitle
\begin{abstract}
In the $\pMaxrSAT{}$ problem, the input is a CNF formula with $n$ variables where each clause is a disjunction of at most $r$ literals. The objective is to compute an assignment which satisfies as many of the clauses as possible. While there are a large number of polynomial-time approximation algorithms for this problem, we take the viewpoint of space complexity following [Biswas et al., Algorithmica 2021] and design sublinear-space approximation algorithms for the problem.

We show that the classical algorithm of [Lieberherr and Specker, JACM 1981] can be implemented to run in $n^{\Oh{1}}$ time while using $\Oh{\log{n}}$ bits of space. The more advanced algorithms use linear or semi-definite programming, and seem harder to carry out in sublinear space. We show that a more recent algorithm with approximation ratio $\sqrt{2}/2$ [Chou et al., FOCS 2020], designed for the streaming model, can be implemented to run in time $n^{\Oh{r}}$ using $\Oh{r \log{n}}$ bits of space. While known streaming algorithms for the problem approximate optimum \textit{values} and use randomization, our algorithms are deterministic and can output the approximately optimal assignments in sublinear space.

For instances of $\pMaxrSAT{}$ with planar incidence graphs, we devise a factor-$(1 - \epsilon)$ approximation scheme which computes assignments in time $n^{\Oh{r / \epsilon}}$ and uses $\max\brc{\sqrt{n} \log{n}, (r / \epsilon) \log^2{n}}$ bits of space.

\end{abstract}

\section{Introduction, Motivation and Our Results}
Starting in the 70's, there has been a long line of work on the approximation properties of $\cNP{}$-hard problems. The classical approach has been to obtain better-than-trivial approximations for such problems with polynomial-time algorithms. Later on, a number of such problems were also studied in the streaming model of computation, where an algorithm must read the input in a fixed (possibly adversarial) sequence. The goal is typically to compute an approximation by making a constant number of passes over the input using space sublinear in the input size. Recently, there has been some interest in studying approximation problems in the sublinear-space RAM model, a model halfway between the RAM and streaming models of computation. In this paper, we continue the work initiated in \cite{BRS2021Algorithmica} and devise sublinear-space approximation algorithms for $\pMaxrSAT{}$.

An instance of $\pMaxrSAT{}$ is a CNF formula $F = C_1 \wedge \dotsb \wedge C_m$, where each of the clauses $C_1, \dotsc, C_m$ is a disjunction of at most $r$ literals over a variable set $\brc{x_1, ..., x_n}$. The objective is to compute an assignment which satisfies as many of the clauses as possible. Viewing the variables and clauses as an incidence structure yields an incidence graph where clauses and variables are vertices, and there is an edge between a variable $x$ and a clause $C$ whenever $x$ appears in $C$. We call the restriction of $\pMaxrSAT{}$ to instances with planar incidence graphs $\pPlMaxrSAT{}$.

The classical approximation algorithm~\cite{Joh1974JCSS} for $\pMaxrSAT{}$ achieves an approximation ratio of $1 / 2$ (shown to be $2 / 3$ in~\cite{CFZ1999JCSS}). Later on, the ratio was improved to $(\sqrt{5} - 1) / 2$ in \cite{LS1981JACM}. Our first observation is that these ratios can be achieved using logarithmic space. Algorithms computing $(3/4)$-approximations are known~\cite{GW1994SIDMA}, but they use linear or semi-definite programming. Under logarithmic-space reductions, it is $\cP{}$-complete to approximate $\pLP{}$ to any constant factor~\cite{Ser1991IPL}. In Section~\ref{sect:general}, we show that the previously mentioned factor-$((\sqrt{5} - 1) / 2)$ and a more recent factor-$(\sqrt{2} / 2)$ approximation algorithm~\cite{CGV2020FOCS}, devised for the streaming model, can be implemented to use $\Oh{\log{n}}$ bits of space.
.

For $\pPlMaxrSAT{}$, it is possible to compute factor-$(1 - \epsilon)$ approximations in polynomial time for any constant $\epsilon > 0$~\cite{KM1996STOC}. In Section~\ref{sect:planar}, we give a sublinear-space implementation of this scheme using recent results about computing tree decompositions~\cite{EJT2010FOCS} and BFS traversal sequences~\cite{AKNW2014MFCS}.

\subsubsection*{The Model.} We use the standard RAM model and additionally constrain the amount of space available to be sublinear in the input size. The input to an algorithm is provided using some canonical representation, which it can read but not modify, i.e.\ it has read-only access to the input. It also has read-write access to a certain amount of auxiliary space. Output is written to a stream: once something is output, the algorithm cannot read it back at a later point as it executes. We count the amount of auxiliary space in single-bit units, and the objective is to use as little auxiliary space as possible.

\subsubsection*{Related Work.}
In the RAM model, earlier works with an emphasis on space efficiency include reachability~\cite{Sav1970JCSS,BBRS1998SICOMP,Rei2008JACM}, sorting and selection~\cite{MP1980TCS,Fre1987JCSS,MR1996TCS} and graph recognition~\cite{Rei1984JACM,AM2004InfComput,EK2014STOC}. In recent years, new results on the computability of separators for planar graphs in sublinear space have been used to devise sublinear-space algorithms for BFS~\cite{AKNW2014MFCS} and DFS~\cite{IO2020ICALP} with better running times than algorithms for general graphs.

\subsubsection*{Results.}
We study the question of what approximations may be achieved when the amount of space available to an algorithm is sublinear in the input size. Our model being more relaxed than the streaming model, we are able to compute approximately optimal assignments for $\pMaxrSAT{}$ instead of approximating optimum values. On the other hand, our model is more restrictive than the RAM model of classical approximation algorithms where the amount of space used by an algorithm can potentially be polynomially large in the input size.

\begin{itemize}
    \item For general $\pMaxrSAT{}$ (Section~\ref{sect:general}), we convert a classical algorithm of Lieberherr and Specker~\cite{LS1981JACM} to our model, obtaining a $((\sqrt{5} - 1) / 2)$-approximation algorithm which uses $\Oh{\log{n}}$ bits of space. We also convert a more recent algorithm of Chou et al.~\cite{CGV2020FOCS} to obtain a $(\sqrt{2} / 2)$-approximation algorithm which uses $\Oh{r \log{n}}$ bits of space.

    \item For $\pPlMaxrSAT{}$ (Section~\ref{sect:planar}), we show how a $(1 - \epsilon)$-approximation scheme of Khanna and Motwani~\cite{KM1996STOC} can be implemented to use $\max\brc{\sqrt{n}, (r / \epsilon) \log{n}}$ bits of space.
\end{itemize}

\section{Preliminaries}
In this paper, we use the following standard notation and concepts. The set $\brc{0, 1, \dotsc}$ of natural numbers is denoted by $\N$ and the set $\brc{1, 2, \dotsc}$ of positive integers is denoted by $\Z^+$. For $n \in \Z^+$, $[n]$ denotes the set $\brc{1, 2, \dotsc, n}$.

An $r$-CNF formula is a conjunction (OR) of disjunctions (AND) of at most $r$ literals (variables or their negations). The individual disjunctions are called clauses of the formula. A clause that consists of a single literal is called a unit clause. For $k \in [n]$, a $k$-clause is a clause which contains exactly $k$ literals.

Let $F$ be a CNF formula with variables $x_1, \dotsc, x_n$. An assignment for $F$ is a function $\phi: [n] \to \brc{0, 1}^n$. The assignment is said to satisfy a clause in $F$ if setting $x_i = \phi{1}\ (i \in [n])$ makes some literal in the clause evaluate to $1$. If $\phi$ satisfies all clauses in $F$, it is said to satisfy $F$. 

\subsection{Time and Space Overheads}\label{ssct:ts_overhead}
In proofs, we measure resource costs in terms of overheads for individual steps. Since the space available to an algorithm is limited, objects created by processing the input are not stored, but recomputed on the fly. For example, consider a procedure (call it $\algo{A}$) that reads an input formula $F$ and produces a subformula $F'$ consisting of the unit clauses of $F$. The procedure outputs $F'$ as a stream $S_{F'}$. Later on, when another procedure (call it $\algo{B}$) reads a portion of $S_{F'}$, $\algo{A}$ recomputes the entire stream $S_{F'}$. Suppose the resource costs of $\algo{A}$ are $t_{\algo{A}}$ time and $s_{\algo{A}}$ space, and assuming $\Oh{1}$-time read costs, suppose the resource costs of $\algo{A}$ are $t_{\algo{B}}$ time and $s_{\algo{B}}$ space.

In this scenario, we call $t_{\algo{B}}$ and $s_{\algo{B}}$ the resource overhead of $\algo{B}$. Combining this overhead with resource costs of $\algo{A}$, we obtain the actual resource costs of $\algo{B}$: $t_{\algo{B}} \cdot t_{\algo{A}}$ time and $s_{\algo{B}} + s_{\algo{A}}$ space.

\subsection{Universal Hash Families}\label{ssct:uhash}
Algorithms appearing later on use the trick of randomized sampling to show that certain good assignments exist and then derandomize the procedure by using a $k$-\emph{universal} family of functions. A $k$-\emph{universal hash} family is a family $\mathcal {H}$ of functions from $[n]$ to $[b]$, for positive integers $n, k, b$ with $n \geq b, k$, such that for random variables $X_i$ ($i \in [n]$) defined as $X_i = f(i)$ with $f$ sampled uniformly at random from $\mathcal{H}$ (denoted $f \sim \mathcal{H}$), the probability---for any $S \subseteq [n]$ with $\abs{S} = k$ and any $a_i \in [l]\ (i \in S)$---of the event $(\bigwedge_{i \in S} X_i = a_i)$ is $1 / b^k$. This condition implies in particular that $X_1, \dotsc, X_n$ are $k$-wise independent and the probability of the event $(X_i = a_i)$ is $1 / b$. 

Let $a \leq b$ be a positive integer, and consider the function $\phi: [b] \to \brc{0, 1}$ defined by $\phi(x) = 1$ if $x \leq a$ and $\phi(x) = 0$ otherwise. With $f \sim \mathcal{H}$ and $Y_1, \dotsc, Y_n$ defined as $Y_i = \phi(f(i))$, it is easy to see that $\Prs{f \sim \mathcal{H}}{Y_i = 1} = a / b$, and by the $k$-universality of $\mathcal{H}$, the variables $Y_1, \dotsc, Y_n$ are $k$-wise independent. Note that $\setb{\phi \circ f}{f \in \mathcal{H}}$ is in fact a $k$-universal hash family. With access to $f \in \mathcal{H}$, the composition $\phi \circ f$ can be computed using $\Oh{\log{n}}$ bits of extra space. 

It is known that $k$-universal hash families such as $\mathcal{H}$ exist~\cite{FKS1984JACM} and can be computed in time $n^{\Oh{k}}$ using $\Oh{k \log{n}}$ bits of space. The following proposition is a combination of those results and the preceding discussion.

\begin{proposition}[Fredman et al.~\cite{FKS1984JACM}]\label{prop:uhash}
	Let $n, k, a, b \in \Z^+$ with $n \geq b \geq a$ and $n \geq k$. One can enumerate a $k$-universal hash family $\oper{Univ}(n, k, a, b)$ for $\brs{[n] \to \brc{0, 1}}$ in time $n^{\Oh{k}}$ using $\Oh{k \log{n}}$ bits of space.	
\end{proposition}

\section{$\pMaxrSAT{}$}\label{sect:general}
In this section, we devise sublinear-space $(\sqrt{5} - 1) / 2)$- and $(\sqrt{2} / 2)$-approximation algorithms for $\pMaxrSAT{}$, with the former's time and space costs being independent of $r$. The following folklore result gives a straightforward linear-time, logarithmic-space $(1 / 2)$-approximation.

\begin{proposition}[Folklore]
    For any $r$-CNF formula, either the all-$1$'s or the all-$0$'s assignment satisfies at least half the clauses.
\end{proposition}

\subsection{Factor-$((\sqrt{5} - 1) / 2)$ Approximation Algorithm}
In what follows, we give a logarithmic-space implementation of the following result.
\begin{proposition}[Lieberherr and Specker~\cite{LS1981JACM}, Theorem 1]
    Let $F$ be an $r$-CNF formula with $m$ clauses. There is an assignment for $F$ which satisfies at least $(\sqrt{5} - 1) m / 2$ clauses.
\end{proposition}

\begin{definition}[$2$-Satisfiability]
	An $r$-CNF formula $F$ is called $2$-satisfiable if any two of its clauses can be simultaneously satisfied, i.e.\ $F$ does not contain a pair $(l, \neg l)$ of literals as clauses.
\end{definition}

The following proposition is based on arguments in \cite{LS1981JACM} (see also~\cite{WS2011book}).
\begin{proposition}\label{prop:gen_gold}
	Let $F$ be a $2$-satisfiable $r$-CNF formula with $m$ clauses in which all unit clauses are positive literals. For the pairwise-independent random assignment where each variable of $F$ is set to $1$ with probability $p = 0.618 \approx (\sqrt{5} - 1) / 2$, the expected number of satisfied clauses is $0.618 m$.
\end{proposition}

We now show how the above proposition can be used to compute $0.618$-approximate optimal $\pMaxrSAT{}$ assignments for general $r$-CNF formulas in logarithmic space.

\begin{theorem}\label{thrm:gen_gold}
	For any instance of $\pMaxrSAT{}$ with $n$ variables, one can compute a $0.618$-approximate optimal assignment in time $n^{\Oh{1}}$ using $\Oh{\log{n}}$ bits of space.
\end{theorem}

\begin{proof}
    Let $F$ be an $r$-CNF formula with variables $x_1, \dotsc, x_n$. In what follows, we describe an algorithm which proves the claim.

    \textbf{Computing an equivalent $2$-satisfiable formula $F'$.} For each clause $C$ in $F$ with at least two literals, check if any variables $x$ appearing in $C$ also appear as a negated clauses $\neg x$ in $F$. If they do, flip the $x$-literals (replace $x$ with $\neg x$ or $\neg x$ with $x$) in $C$ and output the resulting clause. Otherwise, output $C$. The clauses not output yet are unit clauses, i.e.\ they have exactly $1$ literal. For each variable $x_i$, check if $x_i$ appears as a unit clause in $F$. If it does, output $x_i$. Then output the special flag \texttt{\#NEG}, to indicate that clauses to follow appear negated in $F$. For each variable $x_i$, check if it appears as a unit clause $\neg x_i$ in $F$. If it does, check if the unit clause $x_i$ also appears in $F$. If both $\neg x_i$ and $x_i$ are clauses in $F$, output nothing. Otherwise, output $x_i$. Observe that the only clauses of $F$ not output are unit clauses that appear in pairs $(l , \neg l)$.
    
    Let $F'$ be the conjunction of the clauses output and $S_{F'}$ be the stream output. With random access to $F$, $S_F$ is produced in time $n^{\Oh{1}}$ using $\Oh{\log{n}}$ bits of space. Clearly, $F'$ is $2$-satisfiable. Let $\phi$ be an assignment for $F'$. Define $\phi'(x_i) = 1 - \phi(x_i)$ for every $x_i$ appearing after the $\texttt{\#NEG}$ flag in $S_{F'}$ and define $\phi'(x_i) = \phi(x_i)$ otherwise. It is easy to see that $\phi$ satisfies the same number of clauses in $F'$ as $\phi'$ does in $F$, and that given access to $S_{F'}$ and $\phi$, the overhead for computing $\phi'$ is $n^{\Oh{1}}$ time and $\Oh{\log{n}}$ space. We use this transformation later on to compute an assignment for $F$ from an assignment for $F'$.

    \textbf{Computing an assignment for $F'$.} Using the procedure of Proposition~\ref{prop:uhash}, compute a $2$-universal hash family $\mathcal{H} = \oper{Univ}(n, 2, 618, 1000)$ and denote the stream of functions by $S_H$. Note that with $X_i \sim \mathcal{H}_i$ for $i \in [n]$, the random variables $X_1, \dotsc, X_n$ form a pairwise-independent random assignment. Thus, one of the assignments in $S_H$ achieves (for the $2$-satisfiable formula $F'$) the expectation value in Proposition~\ref{prop:gen_gold}.
    
    Let $m'$ be the number of clauses in $F'$. For each assignment $\phi$ in $S_H$, scan $S_F'$ to determine the number $c$ of clauses $\phi$ satisfies. If $c > 0.618 m'$, output $\phi$ and skip to the next step. By Proposition~\ref{prop:uhash}, $\oper{Hash}(n, 2, 618, 1000)$ is computed in time $n^{\Oh{1}}$ and $\Oh{\log{n}}$ bits of space, since $k = 2$ is constant. The overhead of this step is therefore $n^{\Oh{1}}$ time and $\Oh{\log{n}}$ bits of space. Denote the output stream of this step by $S_{\phi}$.

    \textbf{Computing an assignment for $F$.} Now convert the assignment $\phi$ from the previous step to an assignment $\phi'$ (according to the transformation described earlier) as follows. For each $x_i$, scan $S_{\phi}$ to determine the value $v = \phi(x_i)$, and scan $S_{F'}$ to determine if $x_i$ appears after the $\texttt{\#NEG}$ (it was flipped). If it does, output the assignment $\phi'(x_i) = 1 - v$. Otherwise, output the assignment $\phi'(x_i) = v$.  Since $\phi$ satisfies $c \geq 0.618 m'$ clauses in $F'$, $\phi'$ satisfies the same number of clauses in $F$. In particular, it satisfies at least a $0.618$-fraction of the non-unit clauses, and unit clauses that do not appear in $(l, \neg l)$ pairs.
    
    Of the pairs $(l , \neg l)$ of unit clauses appearing in $F$, exactly half are satisfied by any assignment for the variables appearing in them. Now for each $x_i$, scan $S_{\phi}$, to determine if $\phi$ assigns it a value. If it does not, output the assignment $\phi'(x_i) = 1$. Clearly, $\phi'$ now also satisfies exactly half of the unit clauses in $F$ appearing in pairs $(l, \neg l)$, i.e.\ it is an optimal assignment for those clauses. Thus, $\phi'$ is $0.618$-optimal assignment for all of $F$. The overhead of this conversion step is also $n^{\Oh{1}}$ time and $\Oh{\log{n}}$ bits of space. 
    
    Since the overheads for all steps are $n^{\Oh{1}}$ time and $\Oh{\log{n}}$ space, the overall running time is ${(n^{\Oh{1}})}^3 = n^{\Oh{1}}$ and the space used is $3 \cdot \Oh{\log{n}} = \Oh{\log{n}}$.
    \qed
\end{proof}

\subsection{Factor-$(\sqrt{2} / 2)$ Approximation Algorithm}
In the following, we adapt arguments in~\cite{CGV2020FOCS} to devise a $(\sqrt{2} / 2)$-approximation algorithm which runs in time $n^{\Oh{r}}$ and uses $\Oh{r \log{n}}$ bits of space. Consider the following definitions.

\begin{definition}[Bias]\label{defn:bias}
    Let $F$ be an $r$-CNF formula with variables $x_1, \dotsc, x_n$. For $i \in [n]$, the bias of $x_i$ is $\oper{bias}(x_i) =  \sum_{j \in [r]} (\#(j\text{-clauses containing}\ x_i) - \#(j\text{-clauses containing}\ \neg x_i)) / 2^j$.

    The bias of the entire formula is $\oper{bias}(F) = \sum_{i \in [n]} \abs{\oper{bias}(x_i)}$ and the formula $F$ is called positively biased if $\oper{bias}(x_i) \geq 0$ for each $i \in [n]$.
\end{definition}

The next proposition shows that depending on whether the bias of a formula is smaller than a certain value, one can satisfy a good proportion (in expectation) of the clauses in it by setting each variable to $1$ with fixed (bias-dependent) probability.

\begin{proposition}[Chou et al.~\cite{CGV2020FOCS}]\label{prop:gen_sqrt2}
    Let $F$ be a positively-biased $r$-CNF formula with $m$ clauses. For $i \in [r]$, let $m_i$ be the number of $i$-clauses in $F$. The following statements are true.
    \begin{itemize}
        \item The all-$1$'s assignment satisfies at least $\frac{\oper{bias}(F)}{2} + \sum_{i \in [r]} \frac{i m_i}{2^i}$ clauses in $F$.
        \item When $\oper{bias}(F) \leq b^* = 4 \sum_{i \in [r]} \brn{1 - \frac{i + 1}{2^i}} m_i$, an $r$-wise independent random assignment where variables are set to $1$ with probability $\frac{m - \oper{bias}(F)}{2m - 4 \oper{bias}(F)} \leq 1$ satisfies, in expectation, at least $\sum_{i \in [r]} \brn{1 - \frac{1}{2^i}} m_i + \frac{\oper{bias}(F)^2}{4 b^*}$ clauses in $F$.
        \item The best of the two assignments above satisfies at least a $(\sqrt{2} / 2)$-fraction of the maximum number of simultaneously-satisfiable clauses in $F$.
    \end{itemize}
\end{proposition}

We now show how the above proposition can be used to compute good approximations in sublinear space. For any $r$-CNF formula $F$, we first compute an equivalent positively-biased formula $F'$ and then using Proposition~\ref{prop:uhash}, compute an assignment for $F'$ which is a $(\sqrt{2} / 2)$-approximation. We then convert this to an assignment for $F$ satisfying the same number of clauses.

\begin{theorem}\label{thrm:gen_sqrt2}
	For any instance of $\pMaxrSAT{}$ with $n$ variables, one can compute a $(\sqrt{2} / 2)$-approximate optimal assignment in time $n^{\Oh{r}}$ using $\Oh{r \log{n}}$ bits of space.
\end{theorem}

\begin{proof}
    Let $F$ be an $r$-CNF formula with variables $x_1, \dotsc, x_n$ and for $i \in [n]$, let $m_i$ be the number of $i$-clauses in $F$. In what follows, we describe an algorithm which proves the claim.

    \textbf{Computing $\oper{bias}(F)$ and $b^*$.} Set $b_F, b^* \gets 0$. For each $i \in [n]$, compute $b_i = \oper{bias}(x_i)$ and $m_i$. It is easy to see that with random access to $F$, this can be done in logarithmic space. Set $b_F \gets b_F + \abs{b_i},\ b^* \gets b^* + (1 - (i + 1) / 2^i) m_i$, and if $b_i < 0$, output $x_i$ to indicate that $x_i$ has negative bias in $F$. Then discard $(b_i, m_i)$ and move to the next iteration. Finally, store $b_F$ and $b^* \gets 4 b^*$ for later steps using $\Oh{\log{n}}$ bits of space. The entire loop takes time $n^{\Oh{1}}$ and uses $\Oh{\log{n}}$ bits of space. Let $S_B$ be the stream output. 

    \textbf{Computing an equivalent positively-biased formula $F'$.} For each clause $C$ in $F$, check if any variables $x$ appearing in $C$ also appear in the stream $S_B$. If they do, flip the $x$-literals (replace $x$ with $\neg x$ or $\neg x$ with $x$) in $C$ and output the resulting clause. Otherwise, output $C$. Observe that the variables $x$ flipped are precisely those for which $\oper{bias}(x) < 0$ in the previous step. Thus, the clauses output form a positively-biased formula. Denote the output stream by $S_{F'}$. The overhead of this step is $n^{\Oh{1}}$ time and $\Oh{\log{n}}$ space.

    \textbf{Computing an assignment for $F'$.} If $b_F > b^*$, then output the all-$1$'s assignment and skip to the next step. Otherwise, using the procedure of Proposition~\ref{prop:uhash}, compute an $r$-universal hash family $\mathcal{H} = \oper{Univ}(n, r, \ceil{m - b_F}, \ceil{2m - 4 b_F})$ and denote the stream of functions by $S_H$. Similarly as in the proof of Theorem~\ref{thrm:gen_gold}, one of the assignments in $S_H$ achieves the expectation value in Proposition~\ref{prop:gen_sqrt2}.
    
    For each assignment $\phi$ in $S_H$, scan $S_F'$ to determine the number $c$ of clauses $\phi$ satisfies. If $c \geq {b_F}^2 / (16 \sum_{i = 2}^k (1 - (i + 1) / 2^i) m_i)$, output $\phi$ and skip to the next step. The family of assignments is computed in time $n^{\Oh{r}}$ and $\Oh{r \log{n}}$ bits of space, so the overhead of this step is $n^{\Oh{r}}$ time and $\Oh{r \log{n}}$ bits of space. Denote the output stream of this step by $S_{\phi}$.

    \textbf{Computing an assignment for $F$.} Convert the assignment $\phi$ from the previous step to an assignment $\phi'$ for $F$ as follows. For each $x_i$, scan $S_{\phi}$ to determine the value $v = \phi(x_i)$, and scan $S_{B}$ to check if $x_i$ appears in it (it was flipped). If it does, output the assignment $\phi'(x_i) = 1 - v$. Otherwise, output the assignment $\phi'(x_i) = v$. Clearly, $\phi$ satisfies the same number of clauses in $F'$ as $\phi'$ does in $F$. By Proposition~\ref{prop:gen_sqrt2}, this number is at least a $(\sqrt{2} / 2)$-fraction of the maximum number of simultaneously-satisfiable clauses in $F$. With access to $S_{\phi}$ and $S_B$, the overhead of this step is $n^{\Oh{1}}$ time and $\Oh{\log{n}}$ space.  
    
    Thus, the algorithm outputs a $(\sqrt{2} / 2)$-approximate optimal assignment as required. Observe that the maximum overhead of any of the steps is $n^{\Oh{r}}$ time and $\Oh{r \log{n}}$ space. Combining the (constantly many) overheads, the overall running time is $n^{\Oh{r} \cdot \Oh{1}} = n^{\Oh{r}}$ and the space used is $\Oh{r \log{n}} \cdot \Oh{1} = \Oh{r \log{n}}$.
	\qed
\end{proof}

\section{$\pPlMaxrSAT{}$}\label{sect:planar}
In this section, we devise a sublinear-space PTAS for $\pPlMaxrSAT{}$ along the lines of~\cite{KM1996STOC} using the partitioning approach in~\cite{Bak1994JACM} for planar graph problems. We use the following result to perform a BFS traversal of (the incidence graphs of) the input instances in sublinear space.

\begin{proposition}[Chakraborty and Tewari~\cite{CT2015report}, Theorem 1]\label{prop:plan_BFS}
    There is an algorithm which takes as input a planar graph on $n$ vertices and computes a BFS sequence for $G$ in time $n^{\Oh{1}}$ using $\Oh{\sqrt{n} \log{n}}$ bits of space.
\end{proposition}

The next result shows how to use the BFS traversal procedure to partition---in sublinear space---the input formulas into subformulas of bounded diameter.
\begin{lemma}\label{lemm:plan_partition}
    Let $F$ be an $r$-CNF formula with $n$ variables and $m$ clauses that has a planar incidence graph and let $k \in \N$. One can compute a sequence $F_1, \dotsc, F_l$ of subformulas of $F$ such that
    \begin{enumerate}
        \item the diameter of the incidence graph of each $F_i$ ($i \in [l]$) is at most $k$,
		\item $F_i$ and $F_j$ have no variables in common for all $i, j \in [l]$ with $i \neq j$, and
        \item $F_1, \dotsc, F_l$ together contain at least $(1 - 1 / k) m$ clauses of $F$.
    \end{enumerate}
    The procedure runs in time $n^{\Oh{1}}$ and uses $\Oh{\sqrt{n} \log{n}}$ bits of space.
\end{lemma}
\begin{proof}
	Let ${x_1, \dotsc, x_n}$ be the set of variables in $F$, $\brc{C_1, \dotsc, C_m}$ be the set of clauses in $F$, $G_F$ be the incidence graph of $F$, and $V_F$ (resp.\ $C_F$) be the vertices of $G_F$ corresponding to the variables (resp.\ clauses) of $F$. In what follows, we describe a procedure which proves the claim.

	\textbf{Adding a dummy vertex.} This step ensures that $G_F$ is connected. Determine the connected components of $G_F$ using the connectivity algorithm of Asano et al.~\cite{AKNW2014MFCS}: for any two vertices, it runs in time $n^{\Oh{1}}$ and uses $\Oh{\sqrt{n} \log{n}}$ bits of space to check if the two vertices are connected. Then add a dummy variable vertex $x_{n + 1}$ which has an edge to an arbitrary clause vertex in each connected component, making $G_F$ connected. Additionally, add the clause $\neg x_{n + 1}$ (with an edge to $x_{n + 1}$) to ensure that assignments for the formula $F'$ determined by the resulting graph $G_{F'}$ are in $1$-$1$ correspondence with assignments for $F$. Now output $F'$ and $G_{F'}$, and denote this output stream by $S_{F'}$. With random access to $G_F$, it is not hard to see that this transformation runs in time $n^{\Oh{1}}$ and uses $\Oh{\sqrt{n} \log{n}}$ bits of space.

	\textbf{Determining the BFS levels of $G_{F'}$.} 
	Consider a BFS traversal of $G_{F'}$ starting at (the variable vertex corresponding to) $x_{n + 1}$. Suppose the depth of the traversal is $d_0$. Let $d = d_0$ if $d_0$ is even and $d = d_0 + 1$ otherwise. For $i \in [d]$, set $L_i = \setb{v \in V(G_{F'})}{\oper{dist}(u, v) = i - 1}$. Observe that $L_1, \dotsc, L_d$ are precisely the levels of the BFS tree, with $L_i \subseteq V_F$ for odd $i$ and $L_i \subseteq C_F$ for even $i$.

	\textbf{Splitting $G_F$.} Consider the following subsets of $\V{G_{F'}}$.
	\begin{itemize}
		\item For $i \in [d / 2 - 1]$, let $U_i = L_{2i} \cup L_{2i + 1} \cup L_{2i + 2}$. Observe that $U_i \cap U_j \neq \emptyset$ iff $\abs{i - j} \leq 1$ and for $i \in [d / 2 - 1]$, $U_i \cap U_{i + 1} = L_{2i + 2}$.
		
		\item For $i \in \brc{0, \dotsc, k - 1}$, let $W_i = \bigcup_{j \equiv i \pmod{k}} U_j$. Observe that $W_i \cap W_j \neq \emptyset$ iff $i - j \equiv \pm 1 \pmod{k}$ and for $i \in [d / 2 - 1]$, $W_i \cap W_{i + 1} = \bigcup_{j \equiv i \pmod{k}} L_{2j + 2}$.

		\item For any $A \subseteq V_F \cup C_F$, let $C(A)$ be the clause vertices that appear in $A$, i.e.\ $C(A) = A \cap C_F$.
	\end{itemize}

	Clearly, for $i \in [d / 2 - 1]$, $C(W_i) = \bigcup_{j \equiv i \pmod{k}} L_{2j} \cup L_{2j + 2}$ and $C_F = \bigcup_{i \in {0, \dotsc, k - 1}} C(W_i)$. By the inclusion-exclusion principle, we have
    \begin{align*}
        \abs{C(W_0)} + \dotsb &+ \abs{C(W_{k - 1})} = \abs{C_F} + \abs{C(W_0) \cap C(W_1)} + \dotsb + \abs{C(W_{k - 1}) \cap C(W_0)}\\
        &= \abs{C_F} + \sum_{i \in \brc{0, \dotsc k - 1}} \abs{L_{2j + 1}} \leq \abs{C_F} + \abs{C_F} = 2 \abs{C_F}.
	\end{align*}
	
	Thus, for some $i \in \brc{0, \dotsc, k - 1}$, we have $\abs{C(W_i)} \leq 2 \abs{C_F} / k$, i.e.\ $W_i$ contains at most a $(2 / k)$-fraction of the clauses in $F'$ (and $F$). Consider the graph $G_{F'} - W_i$. Observe that $W_i$ comprises groups of $3$ consecutive layers of the BFS traversal, and consecutive groups are $k - 2$ layers apart. Thus, removing $W_i$ from $G_{F'}$ disconnects $G_{F'}$ into connected components which contain at most $k - 2$ layers of the BFS traversal each, i.e.\ their diameters are at most $k - 2$. It follows that the formula $F^+$ corresponding to $G_{F'} - W_i$ satisfies the conditions of the claim. 

	To compute $F^+$, perform the following steps. Using the procedure of Proposition~\ref{prop:plan_BFS}, perform a BFS traversal of the $G_{F'}$ portion of $S_F'$, starting at $x_{n + 1}$. Let $S_B$ be the stream produced by this procedure. The overhead of the procedure is $n^{\Oh{1}}$ time and $\Oh{\sqrt{n} \log{n}}$ bits of space. For each $i \in [k]$, scan $S_B$ to determine the number $\abs{C(W_i)}$ of clauses in $W_i$. For $i$ achieving the smallest $\abs{C(W_i)}$ in the loop, scan $S_B$ and output only the levels (and edges between them) which do not appear in $W_i$. Let $S_{F^+}$ be this output stream. Now scan $S_{F+}$, and for each sequence of consecutive (connected) levels, output the subformula of $F$ induced by those levels. Observe that $S_{F^+}$ is produced by scanning $S_{F'}$ and the final output is produced by scanning $S_{F^+}$. Each scan only involves counting elements in the stream and truncating parts of the stream to produce the output stream. Thus, the overhead of this entire step is $n^{\Oh{1}}$ time and $\Oh{\sqrt{n} \log{n}}$ space.
	
	For the various steps, the maximum overhead is $n^{\Oh{1}}$ time and $\Oh{\sqrt{n} \log{n}}$ bits of space. Thus, combining the overheads for the various steps, the resource costs of the entire algorithm are $n^{\Oh{1}}$ time and $\sqrt{n} \log{n}$ bits of space.
	\qed
\end{proof}

The next two results allow use to compute tree decompositions for incidence graphs of bounded diameter in sublinear space.
\begin{proposition}[Robertson and Seymour~\cite{RS1984JCTB}, Theorem 2.7]\label{prop:local_tw}
	The treewidth of any planar graph with diameter $d$ is at most $3d + 1$.
\end{proposition}

\begin{proposition}[Elberfeld et al.\cite{EJT2010FOCS}, Lemma III.1]\label{prop:tree_dec}
	Let $G$ be a graph on $n$ vertices with treewidth $k \in \N$. One can compute a tree decomposition of width $4k + 1$ for $G$ such that the decomposition tree is rooted, binary and has depth $\Oh{\log{n}}$. The procedure runs in time $n^{\Oh{k}}$ and uses $\Oh{k \log{n}}$ bits of space.
\end{proposition}

\setlength{\textfloatsep}{0.1cm}
\begin{procedure}[h]
\KwIn{$(T, v, \mathcal{B}, \mathcal{P} \psi)$} 

	$max \gets 0, \phi_{max} \gets \emptyset$\;
	\eIf{$v$ has no children in $T$}{
		store $V_v$, the set of variables in $B_v$ that extend $\psi$
		let $\mathcal{A}_v$ be the set of assignments for variables \;
		\ForEach{$\phi \in \mathcal{A}_v$}{
			determine $val$, the number of clauses appearing in $B_v$ that $\phi$ satisfies\;
			\If{$val > max$}{$max \gets val,\ \phi_{max} \gets \phi$}
		}
		\Return{$(max, \phi_{max})$}
	}{
		determine the left child $v_l$ and the right child $v_r$ of $v$ in $T$ if they exist\; 
		store $V_v$, the set of variables in $B_v$, and those in $B_{v_l}$ and $B_{v_r}$ adjacent to clause variables in $B_v$\;
		let $\mathcal{A}_v$ be the set of assignments for $V_v$ that extend $\psi$\;
		\ForEach{$\phi \in \mathcal{A}_v$}{
			$(val_l , \phi_l) \gets \algoi{BdTWMaxSAT}{T, v_l, \mathcal{B}, \phi}$\;
			$(val_r , \phi_r) \gets \algoi{BdTWMaxSAT}{T, v_r, \mathcal{B}, \phi}$\;
			\If{$val_l + val_r > max$}{$max \gets val_l + val_r,\ \phi_{max} \gets \phi_l \cup \phi_r$}
		}
		\Return{$(max, \phi_{max})$}
	}

\caption{BdTWMaxSAT(): find an optimal assignment}\label{proc:BdTWMaxSAT}
\end{procedure}
\setlength{\floatsep}{0.1cm}

We now show how one can solve $\pPlMaxrSAT{}$ exactly on formulas with incidence graphs of bounded diameter.
\begin{lemma}\label{lemm:plan_bd_width}
    Let $F$ be an $r$-CNF formula with $n$ variables that has a planar incidence graph with diameter $k \in \N$. One can compute an assignment for $F$ satisfying the maximum number of clauses in time $n^{\Oh{rk}}$ using $\Oh{rk \log^2{n}}$ bits of space.
\end{lemma}
\begin{proof}
	Let $G$ be the incidence graph of $F$. Since the diameter of $G$ is $k$, its treewidth is at most $3k + 1$ (Proposition~\ref{prop:local_tw}). Consider a tree decomposition $(T, \mathcal{B})$ for $G$ computed by the procedure of Proposition~\ref{prop:tree_dec}. $T$ is the underlying tree (rooted at a vertex $v_r \in \V{T}$) and $\mathcal{B} = \setb{B_v}{v \in \V{T}}$ is the set of bags in the decomposition. By the proposition, the depth of $T$ is $\Oh{\log{n}}$ and its width is at most $4 \cdot (3k + 1) + 1 = \Oh{k}$, i.e.\ $\abs{B_v} = \Oh{k}$ for all $v \in \V{T}$.

	For each $v \in \V{T}$, let $F_v$ be the subformula of $F$ consisting of all clauses appearing in bags of the subtree of $T$ rooted at $v$. Let $V_v$ be the set of variables in $B_v$, and those in the bags of $v$'s children (if they exist) that are adjacent to variables in $B_v$.

	In what follows, we prove that \hyperref[proc:BdTWMaxSAT]{$\algo{BdTWMaxSAT}$}$(T, v_r, \mathcal{B}, \emptyset)$ ($\emptyset$ denotes the empty assignment) computes an assignment for $F$ satisfying the maximum number of clauses. We momentarily assume constant-time access to $G$ and $(T, \mathcal{B})$.
	
	Assume for induction that for any $v \in V{T}$, any assignment $\phi$ for $V_v$ and any child $v_c$ of $v$, that $\algoi{BdTWMaxSAT}{T, v_c, \mathcal{B}, \phi}$ returns an assignment for $F_{v_c}$ which extends $\phi$ and satisfies the maximum number of clauses in $F_{v_c}$ among all such assignments. Now consider a procedure call $\algoi{BdTWMaxSAT}{T, v, \mathcal{B}, \psi}$. The procedure first determines if $v$ has any children. If it does not, then it iterates over all assignments for variables in $B_v$ that extend $\phi$, finds one that satisfies the maximum number of clauses in $F_v$ and returns it. Thus, the procedure is correct in the base case. Since $\abs{V_v} \leq \abs{B_v} = \Oh{k}$, the number of such assignments is $2^{\Oh{k}}$. The call stack stores $\psi$ and $V_v$, so the assignments can be enumerated in time $2^{\Oh{k}} \cdot n^{\Oh{1}}$ using $\Oh{rk \log{n}}$ bits of extra space. Thus, this section of the procedure runs in time $2^{\Oh{k}} \cdot n^{\Oh{1}}$ and uses $\Oh{k \log{n}}$ bits of space. 

	In the other case, i.e.\ $v$ has children, the procedure determines the left and right children of $v$ by scanning $(T, \mathcal{B})$. It then stores $V_v$, which is polynomial-time and uses $\Oh{rk \log{n}}$ bits of space, since each clause in $B_v$ has at most $r$ literals and thus $\abs{V_v} \leq r \cdot \abs{B_v} = \Oh{rk \log{n}}$. The loop iterates over the set $\mathcal{A}_v$ of assignments for $V_v$ that extend $\psi$. The assignments can be enumerated (since $\psi$ and $V_v$ are stored on the call stack) in time $2^{\Oh{rk}} \cdot n^{\Oh{1}}$ using $\Oh{rk \log{n}}$ bits of extra space. Next, the procedure calls itself recursively and stores the tuples returned. Because of the inductive assumption, $\phi_l$ (resp.\ $\phi_r$) extends $\phi$ and satisfies the maximum number of clauses in $F_{v_l}$ (resp.\ $F_{v_r}$) among all such assignments.

	Observe that because $(T, \mathcal{B})$ is a tree decomposition, the variables outside of $V_v$ that $\phi_l$ sets are distinct from the variables outside of $V_v$ that $\phi_r$ sets. Thus, $\phi_l$ and $\phi_r$ do not conflict with each other. In the loop, the procedure finds an extension $\phi$ of $\psi$ such that its extensions $\phi^{max}_l$ and $\phi^{max}_r$, respectively, satisfy the maximum number of clauses in $F_l$ and $F_r$. Overall, $\phi$ is an extension of $\psi$ which satisfies the maximum possible number of clauses in $F_v$. This proves the inductive claim, and thus the procedure is correct.

	We now prove the resource bounds of the procedure (assuming constant-time access to $G$ and $(T, \mathcal{B})$). Observe that in each recursive call, the the individual steps $\Oh{rk \log{n}}$ bits of space and the loops also use $\Oh{rk \log{n}}$ bits of space. Since $T$ has depth $\Oh{\log{n}}$, the depth of the recursion tree is also $\Oh{\log{n}}$, and therefore the call $\algoi{BdTWMaxSAT}{T, r, \mathcal{B}, \emptyset}$ uses a total of $\Oh{rk \log^2{n}}$ bits of space.

	Outside of the recursive calls, the individual steps of the procedure are polynomial-time and the total running time for the other operations in the loops is $2^{\Oh{rk}} \cdot n^{\Oh{1}}$. Thus, if the recursive calls take time $T$, the overall running time of the procedure is $2^{\Oh{rk}} \cdot 2T + 2^{\Oh{rk}} \cdot n^{\Oh{1}}$. Since the depth of the recursion tree is $\Oh{\log{n}}$, this expression solves to $n^{\Oh{rk}}$.

	Now consider the overheads for computing $G$ and $(T, \mathcal{B})$. $G$ is clearly computable in polynomial time and logarithmic space and by Proposition~\ref{prop:tree_dec}, $(T, \mathcal{B})$ is computable in time $n^{\Oh{k}}$ using $\Oh{k \log{n}}$ bits of space. The real resource costs of $\algoi{BdTWMaxSAT}{T, v_r, \mathcal{B}, \emptyset}$ are therefore $n^{\Oh{rk}} \cdot n^{\Oh{k}} = n^{\Oh{rk}}$ time and $\Oh{rk \log^2{n}} + \Oh{k \log{n}} = \Oh{rk \log^2{n}}$ bits of space.
	\qed
\end{proof}

The next theorem combines the previous results to devise a sublinear-space PTAS for $\pPlMaxrSAT{}$.
\begin{theorem}\label{thrm:planar}
    For any $0 < \epsilon < 1$, one can compute $(1 - \epsilon)$-approximate optimal assignments for $\pPlMaxrSAT{}$ in time $n^{\Oh{r / \epsilon}}$ using $\max\brc{\sqrt{n} \log{n}, (r / \epsilon) \log^2{n}}$ bits of space.
\end{theorem}

\begin{proof}
	Consider the following algorithm. Using the procedure of Lemma~\ref{lemm:plan_partition} with $k = \ceil{1 / \epsilon}$, partition $F$ into subformulas $F_1, \dotsc, F_l$. Then for each $i \in [l]$, use the procedure of Lemma~\ref{lemm:plan_bd_width}, compute an exact solution for $F_i$ and output an assignment. In the end, output assignments $x = 0$ for all variables $x$ not appearing in $F_1, \dotsc, F_l$.
	
	Observe that since the partitioning procedure outputs the subformulas as a stream $S_F = F_1, \dotsc, F_l$, each access to $F_i$ costs a single pass over $S_F$, which adds only an $n^{\Oh{1}}$-time, $\Oh{\log{n}}$-space overhead. By Lemma~\ref{lemm:plan_partition}, the partitioning procedure runs in time $n^{\Oh{1}}$ and uses $\Oh{\sqrt{n} \log{n}}$ bits of space. Combining the overhead for access to $F_i$ and the resource bounds from Lemma~\ref{lemm:plan_bd_width}, solving $F_i$ exactly takes time $n^{\Oh{1}} \cdot n^{\Oh{rk}} = n^{\Oh{r / \epsilon}}$ (since $k = \ceil{1 / \epsilon}$) and uses $\Oh{\log{n}} + \Oh{\sqrt{n} \log{n}} + \Oh{rk\log^2{n}} = \max\brc{\sqrt{n} \log{n}, (r / \epsilon) \log^2{n}}$ bits of space. Finally, each $x = 0$ assignment for a variable not appearing in $F_1, \dotsc, F_l$ costs a single pass over $S_F$. It follows that the total resource costs are $n^{\Oh{r / \epsilon}}$ time and $\max\brc{\sqrt{n} \log{n}, (r / \epsilon) \log^2{n}}$ bits of space.

	We now prove the approximation bound. Let $m$ be the number of clauses in $F$. Observe that Lemma~\ref{lemm:plan_partition} guarantees any two subformulas $F_i$ and $F_j$ ($i , j \in l$ with $i \neq j$) have no variables in common, and the subformulas together contain at least $(1 - 1 / k) m \geq (1 - \epsilon) m$ clauses of $F$. Thus, the assignment produced is valid and satisfies at least $(1 - \epsilon) m$ clauses, i.e.\ it is a $(1 - \epsilon)$ approximate optimal $\pPlMaxrSAT{}$ assignment for $F$.
	\qed
\end{proof}

\bibliography{external/references}
\end{document}